\documentclass[runningheads,breaklinks]{llncs}
\usepackage[T1]{fontenc}
\usepackage{graphicx}
\usepackage{color}
\usepackage{orcidlink}
\usepackage[ruled,vlined,titlenumbered,linesnumbered]{algorithm2e}
\SetKwInOut{Input}{Input}\SetKwInOut{Output}{Output}
\DontPrintSemicolon

\usepackage[hyphenbreaks]{breakurl}

\usepackage[utf8]{inputenc}
\usepackage{amssymb}
\usepackage{amsmath,bm}
\usepackage{pgfplots}
\usetikzlibrary{plotmarks, patterns, tikzmark}
\usetikzlibrary{matrix,decorations.pathreplacing,calc}
\usetikzlibrary{arrows,cd,shapes,trees,positioning,arrows.meta,quotes}
\usepackage{hf-tikz}
\usepackage[nolist]{acronym}

\usepackage{mathtools}
\mathtoolsset{showonlyrefs}

\def\ve#1{{\mathchoice{\mbox{\boldmath$\displaystyle #1$}}%
		{\mbox{\boldmath$\textstyle #1$}}%
		{\mbox{\boldmath$\scriptstyle #1$}}%
		{\mbox{\boldmath$\scriptscriptstyle #1$}}}}

\pgfplotsset{
compat=1.17,
mystyle/.style={
    scale only axis,
    width=0.7\columnwidth,
    height=0.5\columnwidth,
    label style={inner sep=0, font=\normalsize},
    tick label style={font=\scriptsize},
    legend style={font=\scriptsize},
    mark size=3,
    major grid style={dashed},
    line width=0.8pt,
    axis line style = thin}
}

\newcommand{\MeKap}{Metzner--Kapturowski}

\newcommand{\Fqm}{\ensuremath{\mathbb F_{q^m}}}

\newcommand{\Fq}{\ensuremath{\mathbb F_{q}}}

\newcommand{\F}{\ensuremath{\mathbb F}}
\newcommand{\NN}{\ensuremath{\mathbb{N}}}
\newcommand{\ZZ}{\ensuremath{\mathbb{Z}}}
\newcommand{\set}[1]{\ensuremath{\mathcal{#1}}}

\newcommand{\Polyring}{\ensuremath{\Fqm[x]}}

\newcommand{\aut}{\ensuremath{\theta}}

\newcommand{\der}{\ensuremath{\delta}}
\newcommand{\SkewPolyring}{\ensuremath{\Fqm[x;\aut,\der]}}
\newcommand{\SkewPolyringZeroDer}{\ensuremath{\Fqm[x;\aut]}}

\newcommand{\opev}[3]{\ensuremath{{#1}(#2)_{#3}}}

\newcommand{\op}[2]{\ensuremath{\mathcal{D}_{#1}(#2)}}
\newcommand{\opexp}[3]{\ensuremath{\mathcal{D}_{#1}^{#3}(#2)}}

\newcommand{\OCompl}[1]{\ensuremath{\mathcal{O}({#1})}}
\newcommand{\OComplTilde}[1]{\ensuremath{\tilde{\mathcal{O}}({#1})}}

\newcommand{\ldiv}{\textsf{leftDivide}}
\newcommand{\keyeq}{\textsf{solveKeyEquation}}
\newcommand{\keyeqapprox}{\textsf{solveKEviaMAB}}

\newcommand{\lapproxbasis}{\textsf{LeftSkewPMBasis}}

\newcommand{\defeq}{:=}

\newcommand{\modr}{\; \mathrm{mod}_\mathrm{r} \;}

\DeclareMathOperator{\Id}{\textrm{Id}}

\DeclareMathOperator{\wt}{wt}
\DeclareMathOperator{\rk}{rk}

\DeclareMathOperator{\diag}{diag}

\DeclareMathOperator{\PR}{Pr}
\newcommand{\failProb}{\ensuremath{\PR_{\textup{fail}}}}

\renewcommand{\vec}[1]{\ve{#1}} 
\newcommand{\mat}[1]{\ensuremath{\bm{#1}}}

\newcommand{\opMoore}[3]{\ensuremath{\mathfrak{M}_{#1}(#2)_{#3}}}

\newcommand{\opVandermonde}[3]{\ensuremath{\mathfrak{m}_{#1}(#2)_{#3}}}
\newcommand{\genNorm}[2]{\ensuremath{\mathcal{N}_{#1}\left(#2\right)}}

\newcommand{\lclm}{\ensuremath{\mathrm{lclm}}}

\renewcommand{\a}{\vec{a}}
\renewcommand{\b}{\vec{b}}
\renewcommand{\c}{\vec{c}}

\newcommand{\e}{\vec{e}}
\newcommand{\f}{\vec{f}}

\renewcommand{\t}{\vec{t}}

\renewcommand{\v}{\vec{v}}
\newcommand{\w}{\vec{w}}
\newcommand{\x}{\vec{x}}
\newcommand{\y}{\vec{y}}

\newcommand{\A}{\mat{A}}
\newcommand{\B}{\mat{B}}

\newcommand{\G}{\mat{G}}
\newcommand{\I}{\mat{I}}
\renewcommand{\H}{\mat{H}}
\renewcommand{\L}{\mat{L}}

\newcommand{\R}{\mat{R}}

\newcommand{\U}{\mat{U}}
\newcommand{\W}{\mat{W}}
\newcommand{\X}{\mat{X}}

\newcommand{\0}{\ensuremath{\mathbf 0}}

\newcommand{\vecbeta}{\ensuremath{\boldsymbol{\beta}}}

\newcommand{\vecxi}{\ensuremath{\boldsymbol{\xi}}}
\newcommand{\vecchi}{\ensuremath{\boldsymbol{\chi}}}
\newcommand{\vecrho}{\ensuremath{\boldsymbol{\rho}}}

\newcommand{\linRS}[4]{\ensuremath{\mathrm{LRS}[#1, #2; #3, #4]}} 

\newcommand{\horIntLinRS}[5]{\ensuremath{\mathrm{HI}\mathrm{LRS}[#1, #2, #3; #4, #5]}} 

\newcommand{\SumRankWeight}{\ensuremath{\wt_{\Sigma R}}}
\newcommand{\SumRankWeightWithN}[1]{\ensuremath{\wt_{\Sigma R, #1}}}

\newcommand{\SumRankDist}{d_{\ensuremath{\Sigma}R}}
\newcommand{\SumRankDistWithN}{d_{\ensuremath{\Sigma}R, \n}}

\newcommand{\RowspaceFqm}[1]{\ensuremath{{\left\langle #1 \right\rangle}_{q^m}}}

\newcommand{\OMul}[1]{\mathcal{M}(#1)}

\newcommand{\mycolorbox}[2]{\colorbox{#1!20!white}{\rule{0pt}{.4\baselineskip} \ensuremath{\scriptsize #2}}}

\newcommand{\myfbox}[1]{\fbox{\ensuremath{#1}}}

\newcommand{\shot}[2]{\ensuremath{{#1}^{(#2)}}}
\newcommand{\subShot}[3]{\ensuremath{{#1}^{(#3)}}_{#2}}
\newcommand{\len}{\ensuremath{n}}
\newcommand{\lenShot}[1]{\ensuremath{\len_{#1}}}

\newcommand{\degConstraint}{\ensuremath{D}}

\newcommand{\intOrder}{\ensuremath{s}}

\newcommand{\shots}{\ensuremath{\ell}}

\DeclareMathOperator{\rdeg}{rdeg}

\DeclareMathOperator{\mpol}{mpol}
\DeclareMathOperator{\intpol}{intpol}

\newcommand{\mpolArgs}[2]{\ensuremath{\mpol_{(#1)_{#2}}}}
\newcommand{\intpolArgs}[3]{\ensuremath{\intpol^{\phantom{(}#3}_{{(#1)_{#2}}}}}

\newcommand{\ESP}{\ensuremath{\sigma}}

\newcommand{\h}{\vec{h}}

\newcommand{\M}{\ve{M}}

\newcommand{\tmax}{t_\mathsf{max}}

\renewcommand{\b}{\vec{b}}
\newcommand{\bmin}{\b_\mathsf{min}}

\newcommand{\n}{\vec{n}}

\begin{acronym}
	\acro{ESP}{error span polynomial}
	\acro{lclm}{least common left multiple}
	\acro{LRS}{linearized Reed--Solomon}
	\acro{LLRS}{lifted linearized Reed--Solomon}
	\acro{ILRS}{interleaved linearized Reed--Solomon}
	\acro{VILRS}{vertically interleaved linearized Reed--Solomon}
	\acro{HILRS}{horizontally interleaved linearized Reed--Solomon}
	\acro{MSRD}{maximum sum-rank distance}
	\acro{RS}{Reed--Solomon}
	\acro{LRPC}{low-rank parity-check}
	\acro{KEM}{key-encapsulation mechanism}
	\acro{RSL}{rank-support-learning}
	\acro{IRS}{interleaved Reed--Solomon}
	\acro{SDP}{syndrome-decoding problem}
	\acro{ISD}{information-set decoding}
	\acro{NIST}{National Institute of Standards and Technology}
	\acro{PQC}{post-quantum cryptography}
	\acro{BIKE}{Bit Flipping Key Encapsulation}
	\acro{HQC}{Hamming Quasi-Cyclic}
	\acro{SIKE}{Supersingular Isogeny Key Encapsulation}
\end{acronym}

\title{Fast Gao-like Decoding of Horizontally Interleaved Linearized Reed--Solomon Codes\thanks{F. Hörmann and H. Bartz acknowledge the financial support by the Federal Ministry of Education and Research of Germany in the programme of ``Souverän. Digital. Vernetzt.'' Joint project 6G-RIC, project identification number: 16KISK022.}}
\titlerunning{Fast Gao-like Decoding of HILRS Codes}

\author{Felicitas Hörmann\inst{1,2} \orcidlink{0000-0003-2217-9753} \and
Hannes Bartz\inst{1} \orcidlink{0000-0001-7767-1513}}

\authorrunning{F. Hörmann and H. Bartz}

\institute{Institute of Communications and Navigation, German Aerospace Center (DLR), Oberpfaffenhofen--Wessling, Germany\\
\email{\{felicitas.hoermann, hannes.bartz\}@dlr.de}
\and
School of Computer Science, University of St. Gallen, St. Gallen, Switzerland}

\begin{document}

\maketitle

\begin{abstract}
	Both horizontal interleaving as well as the sum-rank metric are currently attractive topics in the field of code-based cryptography, as they could mitigate the problem of large key sizes.
	In contrast to vertical interleaving, where codewords are stacked vertically, each codeword of a horizontally $\intOrder$-interleaved code is the horizontal concatenation of $s$ codewords of $\intOrder$ component codes.
	In the case of \ac{HILRS} codes, these component codes are chosen to be \ac{LRS} codes.

	We provide a Gao-like decoder for \ac{HILRS} codes that is inspired by the respective works for non-interleaved \acl{RS} and Gabidulin codes.
	By applying techniques from the theory of minimal approximant bases, we achieve a complexity of $\OComplTilde{\intOrder^{2.373} n^{1.635}}$ operations in $\Fqm$, where $\OComplTilde{\cdot}$ neglects logarithmic factors, $\intOrder$ is the interleaving order and $n$ denotes the length of the component codes.
	For reasonably small interleaving order $\intOrder \ll n$, this is subquadratic in the component-code length $n$ and improves over the only known syndrome-based decoder for \ac{HILRS} codes with quadratic complexity.
	Moreover, it closes the performance gap to vertically interleaved \ac{LRS} codes for which a decoder of complexity $\OComplTilde{\intOrder^{2.373} n^{1.635}}$ is already known.

	We can decode beyond the unique-decoding radius and handle errors of sum-rank weight up to $\frac{\intOrder}{\intOrder + 1} (n - k)$ for component-code dimension $k$.
	We also give an upper bound on the failure probability in the zero-derivation setting and validate its tightness via Monte Carlo simulations.

\end{abstract}

\setcounter{tocdepth}{3}

\acresetall

\section{Introduction}

The American \ac{NIST} started a competition for \ac{PQC} in 2016.
After three rounds, the lattice-based \ac{KEM} CRYSTALS-Kyber~\cite{crystals-kyber-round3} was standardized in July 2022~\cite{nist-round3}.
Moreover, \ac{NIST} announced a fourth round to which four \ac{KEM} candidates advanced: \acs{BIKE}~\cite{bike-round4}, Classic McEliece~\cite{classic-mceliece-round4}, \acs{HQC}~\cite{hqc-round4}, and \acs{SIKE}~\cite{sike-round4}.
\acs{SIKE} is the only candidate based on hard problems in the area of isogenies and was broken by~\cite{castryck2022sike-attack} shortly after \ac{NIST}'s round-4 announcement.
The remaining three candidates in this round rely on coding-theoretical problems in the Hamming metric.

In his seminal paper~\cite{McEliece-1978} in 1978, McEliece proposed the first code-based cryptosystem, which still serves as a blueprint for most of the recent proposals.
The McEliece framework essentially resisted the cryptanalytic effort of 45 years.
However, it suffers from large key sizes and is thus not usable in many practical applications.

\paragraph{Rank and Sum-Rank Metric}
As the \acl{SDP} in the rank metric is harder than its Hamming-metric counterpart~\cite{Aragon2018-DecAttack,bardet2021rsl}, many McEliece-like schemes based on rank-metric codes as e.g.~\cite{gabidulin1991ideals,gabidulin2009improving,loidreau2010designing,Loidreau-GPT-ACCT2016} were considered.
Unfortunately, most of them were broken by structural attacks.
A new approach is to consider the sum-rank metric which covers both the Hamming and the rank metric as special cases.
Even though the gain in terms of key size might not be as large as for the rank metric, it is reasonable to hope that rank-metric attacks cannot be adapted to the sum-rank-metric case~\cite{Hoermann_Bartz_et-al-2023} and the corresponding systems will remain secure.

\paragraph{Interleaved Codes}
Another way to reduce the key size is to use codes with higher error-correction capability.
An increased error weight will result in higher complexities for generic attacks like~\cite{puchinger2020generic} and thus require smaller parameter sizes to achieve the same level of security.
One well-known code construction to improve the (burst) error-correction capability is interleaving, where each codeword of the $\intOrder$-interleaved code consists of $\intOrder$ vertically or horizontally stacked codewords of $\intOrder$ component codes, respectively.

Metzner and Kapturowski~\cite{metznerkapturowski1990} showed that vertically interleaved Hamming-metric codes can be efficiently decoded with negligible failure probability as soon as their interleaving order $\intOrder$ is high compared to the error weight $t$.
This result was generalized to the rank metric~\cite{metznerkapturowskirank2021isit,puchinger2019decoding} and recently also to the sum-rank metric~\cite{jerkovits2023metzner}.
As no knowledge about the code structure is needed for \MeKap-like decoders, this is a direct generic attack on any code-based cryptosystems based on vertically interleaved codes with high interleaving order.
Thus, horizontal interleaving appears to be better suited for cryptographic purposes.
This is also reflected in recent proposals as for example in the \ac{KEM} LowMS~\cite{aragon2022lowms} that is based on horizontally interleaved Gabidulin codes, in the signature scheme Durandal~\cite{aragon2019durandal} based on the closely related \ac{RSL} problem~\cite{bardet2021rsl}, and in the cryptosystem~\cite{aguilar2022lrpcMultipleSyndromes} that makes use of horizontally interleaved \ac{LRPC} codes~\cite{renner2019efficient}.

The cryptanalysis of the underlying hard problems ensures reliable security-level estimates.
However, also performance improvements for decoding horizontally interleaved codes have a significant impact as they directly speed up decryption and verification within the corresponding cryptosystems and digital signatures.

\paragraph{\acs{HILRS} Codes}
\Ac{HILRS} codes combine the usage of an alternative decoding metric for higher generic-decoding complexity and the interleaving construction for higher error-correction capability.
Both approaches promise to reduce the key size in a McEliece-like setup.
The component codes of an \ac{HILRS} code are \ac{LRS} codes which were introduced by Mart{\'\i}nez-Pe{\~n}as in 2018~\cite{martinez2018skew}.
Up to now, \ac{LRS} codes are one of the most studied code families in the sum-rank metric.
They are evaluation codes with respect to skew polynomials and form the natural generalization of \ac{RS} codes in the Hamming metric and Gabidulin codes in the rank metric.

As the performance of code-based cryptosystems strongly depends on the decoding speed for the underlying codes, fast decoders for \ac{HILRS} codes are crucial.
Currently, the only known decoder for \ac{HILRS} codes is syndrome-based and has a quadratic complexity in the length $sn$ of the interleaved code (ongoing work~\cite{hoermann2023errorerasure} extending~\cite{hoermann2022errorErasureISIT}).
It can handle a combination of errors, row erasures, and column erasures.

In contrast, \ac{VILRS} codes, which are constructed by vertically stacking $\intOrder$ \ac{LRS} codewords, allow for decoding with lower complexity $\tilde{\mathcal{O}}(\intOrder^{\omega} \mathcal{M}(n)) \subseteq \tilde{\mathcal{O}}(\intOrder^{2.373} n^{1.635})$~\cite{bartz2021decoding,bartz2023fast}.
Here, $\omega$ and $\mathcal{M}(n)$ denote the matrix-multiplication coefficient and the cost of multiplying two skew polynomials of degree at most $n$, respectively, and $\tilde{\mathcal{O}}(\cdot)$ neglects logarithmic factors.

\paragraph{Contributions}
This paper presents a Gao-like decoder for \ac{HILRS} codes.
It is based on the original Gao decoder for Reed--Solomon codes in the Hamming metric~\cite{gao2003new} as well as on its known extensions to Gabidulin codes~\cite{wachter2013fast,wachter2013phd} and their horizontally interleaved version~\cite{puchinger2017row} in the rank metric.
We consider probabilistic unique decoding beyond the unique-decoding radius and derive an upper bound on the decoding-failure probability in the zero-derivation case.
We achieve a decoding radius of $\frac{\intOrder}{\intOrder + 1} (n - k)$ for the interleaving order $\intOrder$ and for $n$ and $k$ denoting the length and the dimension of the component codes, respectively.

We further show how a major speedup can be obtained by using the theory of minimal approximant bases~\cite{bartz2021orderBases}.
The fast variant of the Gao-like decoder achieves subquadratic complexity in the length $n$ of the component codes for a fixed interleaving order $s$.
Particularly, we obtain $\tilde{\mathcal{O}}(\intOrder^{\omega} \mathcal{M}(n)) \subseteq \tilde{\mathcal{O}}(\intOrder^{2.373} n^{1.635})$ and thus close the performance gap with respect to the decoding of \ac{VILRS} codes.

Our conceptually new approach to solving the Gao-like key equation results in the fastest known decoder for \ac{HILRS} codes in the sum-rank metric.
Moreover, the special case obtained for the rank metric yields the fastest decoder for horizontally interleaved Gabidulin codes in the rank metric, improving on~\cite{sidorenko2010decoding,sidorenko2011skew,puchinger2017row}.

\paragraph{Outline}
We start the paper in~\autoref{sec:prelims} by giving basic preliminaries on skew polynomials, on \ac{HILRS} codes in the sum-rank metric, and on the channel model we consider.
Then, we present a Gao-like decoder for \ac{HILRS} codes in~\autoref{sec:gao-dec} and analyze its decoding radius, complexity, and failure probability.
\autoref{sec:gao-mab} deals with a speedup for the shown decoder that is based on the theory of minimal approximant bases.
Finally, we summarize the main results of the paper in~\autoref{sec:conclusion} and give an outlook on future work.

\section{Preliminaries}
\label{sec:prelims}

We denote the finite field of order $q$ by $\Fq$ and refer to its degree-$m$ extension field by $\Fqm$.
We often consider vectors $\x \in \Fqm^{n}$ that are divided into blocks.
More precisely, we define a \emph{length partition} of $n \in \NN^{\ast}$ as the vector $\n = (n_1, \dots, n_{\shots}) \in \NN^{\shots}$ with $\sum_{i = 1}^{\shots} n_i = n$ and $n_i > 0$ for all $i = 1, \dots, \shots$.
We write $\x = ( \shot{\x}{1} \mid \dots \mid \shot{\x}{\shots} )$, where the blocks $\shot{\x}{i}$ belong to $\Fqm^{n_i}$ for all $i = 1, \dots, \shots$.
Similarly, we write $\X = ( \shot{\X}{1} \mid \dots \mid \shot{\X}{\shots})$ for a subdivided matrix $\X \in \Fqm^{k \times n}$ with $\shot{\X}{i} \in \Fqm^{k \times n_i}$ for all $i = 1, \dots, \shots$.
The $\Fqm$-linear row space of $\X$ is denoted by $\RowspaceFqm{\X}$.

Further choose an $\Fqm$-automorphism $\aut$ with fixed field $\Fq$.
Note that $\aut$ is $\Fq$-linear and satisfies both $\aut(a + b) = \aut(a) + \aut(b)$ and $\aut(a \cdot b) = \aut(a) \cdot \aut(b)$ for arbitrary $a, b \in \Fqm$.
Moreover, we consider a map $\der: \Fqm \to \Fqm$ for which the equalities $\der(a + b) = \der(a) + \der(b)$ and $\der(ab) = \der(a)b + \aut(a)\der(b)$ hold for all $a, b \in \Fqm$.
In the finite-field setting, all such $\aut$-derivations $\der$ are inner derivations~\cite{martinez2018skew}, i.e., they have the form $\der = \gamma (\Id - \aut)$ for a parameter $\gamma \in \Fqm$ and the identity $\Id$.

The automorphism $\aut$ and the derivation $\der$ give rise to a partition of $\Fqm$ with respect to $(\aut, \der)$-conjugacy~\cite{lam1988vandermonde}.
Namely, two elements $a, b \in \Fqm$ are conjugate if there is a nonzero $c \in \Fqm^{\ast}$ with
\begin{equation}
	a^c \defeq \aut(c) a c^{-1} + \der(c) c^{-1}.
\end{equation}
The conjugacy class of an element $a \in \Fqm$ is denoted by $\set{C}(a) \defeq \left\{a^c : c \in \Fqm^{\ast} \right\}$ and $\set{C}(0)$ is called the trivial conjugacy class.
There are $q - 1$ distinct nontrivial $(\aut, \der)$-conjugacy classes.
In the zero-derivation case, each of the first $q - 1$ powers of any primitive element of $\Fqm$ belongs to another nontrivial class.

\subsection{Skew-Polynomial Rings}

Skew polynomials were first studied by Ore in 1933~\cite{Ore1933,ore1933theory} and are used e.g. for the construction of \ac{LRS} codes~\cite{martinez2018skew}.
The skew-polynomial ring $\SkewPolyring$ contains all formal polynomials $\sum_{i} f_{i} x^{i-1}$ with finitely many nonzero coefficients $f_{i} \in \Fqm$.
The notion of the degree $\deg(f) \defeq \max\{i - 1 : f_i \neq 0\}$ of a skew polynomial $f(x) = \sum_{i} f_{i} x^{i-1}$ carries over from $\Polyring$.
The set of skew polynomials forms a non-commutative ring with respect to conventional polynomial addition and a multiplication that is determined by the non-commutative rule $x a = \aut(a) x + \der(a)$ for any $a \in \Fqm$.
By $\SkewPolyring_{<k}$ we denote the subset of $\SkewPolyring$ containing all skew polynomials of degree less than $k$.
For simplicity, we refer to the skew-polynomial ring with zero derivation by $\SkewPolyringZeroDer \defeq \Fqm[x;\aut,0]$.

$\SkewPolyring$ is Euclidean which ensures the existence of skew polynomials $q, r \in \SkewPolyring$ with $f(x) = q(x) g(x) + r(x)$ and $\deg(r) < \deg(g)$ for each pair $f, g \in \SkewPolyring$ with $\deg(f) \geq \deg(g)$.
We denote the remainder $r$ of this right-hand division by $f \modr g$.
\\

The literature provides two meaningful ways to evaluate skew polynomials, namely, the remainder evaluation~\cite{lam1988vandermonde} and the generalized operator evaluation~\cite{martinez2018skew}.
The former corresponds to the idea of enforcing a remainder theorem similar to the one in conventional polynomial rings and will not be of interest for this paper.
The latter is e.g.\ used for the construction of \ac{LRS} codes that we heavily rely on.
For defining the generalized operator evaluation of skew polynomials we first introduce the operator $\op{a}{b} \defeq \aut(b) a + \der(b)$ and its $i$-th power $\opexp{a}{b}{i} \defeq \op{a}{\opexp{a}{b}{i-1}}$ for $i \in \NN^{\ast}$ and any $a, b \in \Fqm$.
The operator simplifies to $\op{a}{b} = \aut(b) a$ for all $a, b \in \Fqm$ in the case of zero derivation.
In this case, its $i$-th power $\opexp{a}{b}{i}$ for $i \in \NN^{\ast}$ can be written as $\opexp{a}{b}{i} = \aut^i(b) \cdot \genNorm{i}{a}$, where $\genNorm{i}{a} \defeq \prod_{k = 0}^{i - 1} \aut^{k}(a)$ is the $i$-th truncated norm of $a$.

The \emph{generalized operator evaluation} of a skew polynomial $f(x) = \sum_{i = 1}^{d} f_{i} x^{i-1} \allowbreak \in \SkewPolyring$ at a point $b \in \Fqm$ and with respect to an evaluation parameter $a \in \Fqm$ is defined as
\begin{equation}
    \opev{f}{b}{a} \defeq \sum_{i = 1}^{d} f_{i} \opexp{a}{b}{i-1}.
\end{equation}
We use the notation $\opev{f}{\b}{a} \defeq ( \opev{f}{b_1}{a}, \dots, \opev{f}{b_{n}}{a} )$ to denote the vector containing the evaluations of $f$ at every entry of $\b \in \Fqm^{n}$.
Moreover, if $\b = (\shot{\b}{1} \mid \dots \mid \shot{\b}{\shots}) \in \Fqm^{n}$ is subdivided according to a length partition $\n$ and $\a = (a_1, \dots, a_{\shots}) \in \Fqm^{\shots}$, we use the shorthand $\opev{f}{\b}{\a} \defeq ( \opev{f}{\shot{\b}{1}}{a_1}, \dots, \opev{f}{\shot{\b}{\shots}}{a_{\shots}} )$ to evaluate $f$ at the elements of the $i$-th block $\shot{\b}{i}$ with respect to the evaluation parameter $a_i$ for every $i = 1, \dots, \shots$.

The evaluation of a product of two skew polynomials $f,g \in \SkewPolyring$ satisfies the product rule
$\opev{(f\cdot g)}{b}{a} = \opev{f}{\opev{g}{b}{a}}{a}$ for all $a,b \in \Fqm$~\cite{lam1988vandermonde}.

For a vector $\x = \left(\x^{(1)} \mid \dots \mid \x^{(\shots)}\right) \in \Fqm^{n}$, a vector $\a \in \Fqm^{\shots}$, and a parameter $d \in \NN^{\ast}$ the \emph{generalized Moore matrix} $\opMoore{d}{\x}{\a}$ is defined as
\begin{align}\label{eq:def_gen_moore_mat}
    \opMoore{d}{\x}{\a} &\defeq
    \left( \opVandermonde{d}{\x^{(1)}}{a_1} \mid \dots \mid \opVandermonde{d}{\x^{(\shots)}}{a_\shots} \right)
    \in \Fqm^{d \times n},
    \\
    \text{with }
    \opVandermonde{d}{\x^{(i)}}{a_i} &\defeq
    \begin{pmatrix}
        x^{(i)}_1 & \cdots & x^{(i)}_{\lenShot{i}}
        \\
        \op{a_i}{x^{(i)}_1} & \cdots & \op{a_i}{x^{(i)}_{\lenShot{i}}}
        \\
        \vdots & \ddots & \vdots
        \\
        \opexp{a_i}{x^{(i)}_1}{d-1} & \cdots & \opexp{a_i}{x^{(i)}_{\lenShot{i}}}{d-1}
    \end{pmatrix}
    \quad \text{for all } i = 1, \dots, \shots.
\end{align}
If $\a$ contains representatives of pairwise distinct nontrivial conjugacy classes of $\Fqm$ and $\rk_{q}\left(\x^{(i)}\right) =
\lenShot{i}$ for all $i = 1, \dots, \shots$, it holds
$\rk_{q^m}\left(\opMoore{d}{\x}{\a}\right) = \min(d, \len)$~\cite{lam1988vandermonde,martinez2018skew}.

Consider $\b = (\shot{\b}{1} \mid \dots \mid \shot{\b}{\shots}) \in \Fqm^{n}$ and $\a = (a_1, \dots, a_{\shots}) \in \Fqm^{\shots}$.
The minimal skew polynomial that vanishes on the entries of $\shot{\b}{i}$ with respect to the evaluation parameter $a_i$ for each $i = 1, \dots, \shots$ is denoted by
$\mpolArgs{\b}{\a}(x)$ and characterized by
\begin{equation}
  \opev{\mpolArgs{\b}{\a}}{\shot{\b}{i}}{a_i} = \0 \quad \text{for all } i=1,\dots,\shots.
\end{equation}
According to~\cite{Boucher-2020}, it can be computed as a \ac{lclm} via
\begin{equation}
	\label{eq:mpol-lclm}
  	\mpolArgs{\b}{\a}(x)
	= \lclm\left\{ x-\frac{\op{a_i}{\subShot{b}{\iota}{i}}}{\subShot{b}{\iota}{i}} ~:~\subShot{b}{\iota}{i}\neq 0,~~
	\begin{aligned}
		\iota &= 1, \dots, n_i,\\
		i &= 1, \dots, \shots
		\end{aligned}
	\right\}.
\end{equation}
The degree satisfies $\deg(\mpolArgs{\b}{\a})\leq n$ with equality if and only if the entries of $\shot{\b}{i}$ are $\Fq$-linearly independent for all $i = 1, \dots, \shots$ and the evaluation parameters $a_1, \dots, a_{\shots}$ belong to distinct nontrivial conjugacy classes of $\Fqm$.

Now consider an additional vector $\c = (\shot{\c}{1} \mid \dots \mid \shot{\c}{\shots}) \in \Fqm^n$.
Then there exists a unique skew interpolation polynomial $\intpolArgs{\b}{\a}{\c}(x) \in \SkewPolyring$ with $\deg(\intpolArgs{\b}{\a}{\c}) < n$ and
\begin{equation}
	\opev{\intpolArgs{\b}{\a}{\c}}{\shot{\b}{i}}{a_i} = \shot{\c}{i} \quad \text{for all } i=1,\dots,\shots\text{~\cite{caruso2019residues}}.
\end{equation}

For the complexity analysis of the Gao-like decoder, we will use $\OCompl{\cdot}$ to state asymptotic costs in terms of the usual big-O notation.
Moreover, the notation $\OComplTilde{\cdot}$ indicates that logarithmic factors in the input parameter are neglected.
The complexity of skew-polynomial operations in the zero-derivation setting was summarized in~\cite[Section II.D.]{bartz2021orderBases}.
Particularly, left and right division of skew polynomials with degree at most $n$ as well as the computation of a minimal or an interpolation polynomial of degree at most $n$ can be achieved in $\OComplTilde{\mathcal{M}_{q, m}(n)}$ operations in $\Fqm$.
Here, $\mathcal{M}_{q, m}(n)$ denotes the cost of multiplying two skew polynomials of degree $n$ from $\SkewPolyringZeroDer$ and it holds $\OCompl{\mathcal{M}_{q, m}(n)} \subseteq \OCompl{n^{\min(\frac{\omega + 1}{2}, 1.635)}} \subseteq \OCompl{n^{1.635}}$.
The exponent $\omega \geq 2$ denotes the matrix-multiplication coefficient for which the currently best known upper bound is $\omega < 2.3728639$~\cite{le2014powers}.

\subsection{The Sum-Rank Metric and the Corresponding Interleaved Channel Model}

The \emph{sum-rank weight} of a vector $\x = ( \shot{\x}{1} \mid \dots \mid \shot{\x}{\shots} )\in \Fqm^{n}$ with respect to the length partition $\n$ is
\begin{equation}
	\SumRankWeightWithN{\n}(\x) = \sum_{i=1}^{\shots} \rk_{q} \big( \shot{\x}{i} \big)
\end{equation}
where $\rk_{q} \big( \shot{\x}{i} \big)$ is the maximum number of $\Fq$-linearly independent entries of the block $\shot{\x}{i}$ for each $i = 1, \dots, \shots$.
The \emph{sum-rank metric} is induced by the sum-rank weight via $\SumRankDistWithN(\x, \y) = \SumRankWeightWithN{\n}(\x - \y)$ for all vectors $\x,\y \in \Fqm^{n}$.
Note that we omit the index $\n$ and simply write $\SumRankWeight$ and $\SumRankDist$ when the length partition is clear from the context.

The sum-rank metric coincides with the Hamming metric for $\shots = n$, i.e., when every block has length one, and with the rank metric for $\shots = 1$, i.e., when the vector is considered as a single block.
\\

Let now $\x = (\x_1 \mid \dots \mid \x_{\intOrder}) \in \Fqm^{\intOrder n}$ with $\x_j \in \Fqm^n$ for all $j = 1, \dots, \intOrder$ be a horizontally $\intOrder$-interleaved vector for an interleaving order $s \in \NN^{\ast}$.
Let us further assume for simplicity that all component vectors $\x_j = (\subShot{\x}{j}{1} \mid \dots \mid \subShot{\x}{j}{\shots}) \in \Fqm^{n}$ for $j = 1, \dots, \intOrder$ are equipped with the same length partition $\n$.
The natural way to define the sum-rank weight of $\x \in \Fqm^{\intOrder n}$ is with respect to the \emph{block-ordered} length partition $\tilde{\n} = (\intOrder n_{1}, \dots, \intOrder n_{\shots})$, i.e., as
\begin{equation}
	\SumRankWeightWithN{\tilde{\n}}(\x) \defeq \sum_{i=1}^{\shots} \rk_q(\shot{\x}{i}) \qquad \text{for } \shot{\x}{i} = (\subShot{\x}{1}{i} \mid \dots \mid \subShot{\x}{\intOrder}{i}).
\end{equation}
As for the conventional sum-rank metric, we often omit the length partition in the index and simply write $\SumRankWeight(\x)$ when $\tilde{\n}$ is clear from the context.
\autoref{fig:interleaved-sum-rank-weight} illustrates how the sum-rank weight of horizontally interleaved vectors is computed by grouping the same-indexed blocks of the component vectors.
It shows how the block-ordered length partition arises naturally in this setting.

\begin{figure}[ht]
    \begin{align}
		\resizebox{.82\textwidth}{!}{
		$\x =
		\Big(
			\underbrace{\myfbox{\mycolorbox{blue}{\subShot{\x}{1}{1}} ~ \mycolorbox{red}{\subShot{\x}{1}{2}} \cdots \mycolorbox{green}{\subShot{\x}{1}{\shots}}}}_{\x_1 \in \Fqm^{n}}~
			\underbrace{\myfbox{\mycolorbox{blue}{\subShot{\x}{2}{1}} ~ \mycolorbox{red}{\subShot{\x}{2}{2}} \cdots \mycolorbox{green}{\subShot{\x}{2}{\shots}}}}_{\x_2 \in \Fqm^{n}}~
			\cdots~
			\underbrace{\myfbox{\mycolorbox{blue}{\subShot{\x}{\intOrder}{1}} ~ \mycolorbox{red}{\subShot{\x}{\intOrder}{2}} \cdots \mycolorbox{green}{\subShot{\x}{\intOrder}{\shots}}}}_{\x_{\intOrder} \in \Fqm^{n}}
		\Big) \in \Fqm^{\intOrder n}$
		}
		\\
		\resizebox{\textwidth}{!}{
		$\SumRankWeight(\x) =
		\rk_q \Big( \myfbox{\mycolorbox{blue}{\subShot{\x}{1}{1}} ~ \mycolorbox{blue}{\subShot{\x}{2}{1}} \cdots \mycolorbox{blue}{\subShot{\x}{\intOrder}{1}}} \Big)
		+ \rk_q \Big( \myfbox{\mycolorbox{red}{\subShot{\x}{1}{2}} ~ \mycolorbox{red}{\subShot{\x}{2}{2}} \cdots \mycolorbox{red}{\subShot{\x}{\intOrder}{2}}} \Big)
		+ \ldots
		+ \rk_q \Big( \myfbox{\mycolorbox{green}{\subShot{\x}{1}{\shots}} ~ \mycolorbox{green}{\subShot{\x}{2}{\shots}} \cdots \mycolorbox{green}{\subShot{\x}{\intOrder}{\shots}}} \Big)$
		}
	\end{align}
	\vspace{-.5cm}
	\caption{Illustration of the sum-rank weight for a horizontally $\intOrder$-interleaved vector $\x = (\x_1 \mid \dots \mid \x_{\intOrder}) \in \Fqm^{\intOrder n}$.}
	\label{fig:interleaved-sum-rank-weight}
\end{figure}

We now consider the transmission of an interleaved vector $\x \in \Fqm^{\intOrder n}$ over a sum-rank error channel with output
\begin{equation}
	\label{eq:sum_rank_channel_hor_int}
	\y = \x + \e
\end{equation}
where the error vector $\e$ is understood as a horizontally $\intOrder$-interleaved vector $\e = (\e_1 \mid \dots \mid \e_{\intOrder}) \in \Fqm^{\intOrder n}$ of sum-rank weight $\SumRankWeightWithN{\tilde{\n}}(\e) = t$.
We further assume a uniform channel distribution, that is, that the error $\e$ is drawn uniformly at random from the set
\begin{equation}
	\label{eq:e-set-random}
    \{ \x = (\x_1 \mid \dots \mid \x_{\intOrder}) \in \Fqm^{\intOrder n} : \SumRankWeightWithN{\tilde{\n}}(\x) = t \}.
\end{equation}
The described channel is illustrated in~\autoref{fig:channel}.

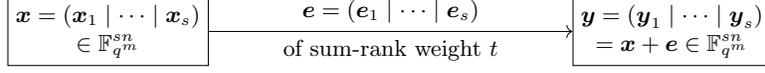
\begin{figure}[ht!]
    \begin{center}
		\resizebox{.85\textwidth}{!}{
		\begin{tikzpicture}
			\node [draw, align=center] (codeword) {$\x = (\x_1 \mid \dots \mid \x_{\intOrder})$\\ $\in \Fqm^{\intOrder n}$};
			\node [draw, align=center, node distance=5cm, right=of codeword] (received) {$\y = (\y_1 \mid \dots \mid \y_{\intOrder})$\\ $= \x + \e \in \Fqm^{\intOrder n}$};
			\draw [->] (codeword) -- node [above] {$\e = (\e_1 \mid \dots \mid \e_{\intOrder})$} node [below] {of sum-rank weight $t$} (received);
		\end{tikzpicture}
		}
	\end{center}
	\caption{The additive sum-rank channel for horizontally interleaved vectors.}
	\label{fig:channel}
\end{figure}

Let $\t = (t_1, \dots, t_{\shots}) \in \NN^{\shots}$ with $t_i = \rk_q(\shot{\e}{i}) \defeq \rk_q(\subShot{\e}{1}{i} \mid \dots \mid \subShot{\e}{\intOrder}{i})$ for all $i = 1, \dots, \shots$ denote the rank partition of $\e$.
Then, we obtain for each $i = 1, \dots, \shots$ a decomposition of the form
    $(\e_1^{(i)} \mid \dots \mid \e_\intOrder^{(i)}) = \shot{\a}{i} \cdot \left( \shot{\B}{i}_{1} \mid \dots \mid \shot{\B}{i}_{\intOrder} \right)$,
where $\a^{(i)}\in\Fqm^{t_i}$ with $\rk_q(\a^{(i)})=t_i$ and $\B_j^{(i)}\in \Fq^{t_i \times n_i}$ with $\rk_q\left(\B_1^{(i)} \mid \dots \mid \B_\intOrder^{(i)}\right)=t_i$ for all $j = 1, \dots, \intOrder$.
After reordering the components, the error vector $\e$ can thus be decomposed as
\begin{equation}\label{eq:err_int_vec_decomp}
    \e
    =\a \cdot \B
\end{equation}
with $\a=(\shot{\a}{1} \mid \dots \mid \shot{\a}{\shots}) \in \Fqm^t$ and
\begin{equation}\label{eq:def_B_matrix}
    \B =
    \left(
    \begin{array}{ccc|c|ccc}
     \B_1^{(1)} & & & & \B_\intOrder^{(1)} &
     \\
     & \ddots & & \dots & & \ddots &
     \\
    & & \B_1^{(\shots)} & &  & & \B_\intOrder^{(\shots)}
    \end{array}
    \right) \in \Fq^{t \times \intOrder n},
\end{equation}
where for any $i = 1, \dots, \shots$ and any $j = 1, \dots, \intOrder$
\begin{align*}
    &\a^{(i)}\in\Fqm^{t_i} \text{ with } \rk_q(\a^{(i)})=t_i
    \\
    \text{and }
    &\B_j^{(i)}\in \Fq^{t_i \times n_i} \text{ with } \rk_q\left(\B_1^{(i)} \mid \dots \mid \B_\intOrder^{(i)}\right)=t_i.
\end{align*}
Note that the decomposition in~\eqref{eq:err_int_vec_decomp} is not unique.
Moreover, the uniform distribution of $\e$ among all vectors of sum-rank weight $t$ implies that, for fixed rank partition $\t$, both $\a$ and $\B$ are also chosen uniformly at random from the sets
\begin{gather}
    \{ \x \in \Fqm^{t} : \SumRankWeightWithN{\t}(\x) = t \} \label{eq:a-set-random}
	\\
	\text{and }\quad \{ \X \in \Fqm{t \times \intOrder n} \text{ of the form}~\eqref{eq:def_B_matrix} : \SumRankWeightWithN{\tilde{\n}}(\X) = t \},
	\label{eq:B-set-random}
\end{gather}
respectively.

The elements in $\shot{\a}{i}$ form a basis of the column space of $\shot{\e}{i}$ and are called \emph{error values}.
Similarly, the rows of $\subShot{\B}{j}{i}$ form a basis of the row space of $\subShot{\e}{j}{i}$ and are referred to as \emph{error locations}.
For horizontal interleaving, the error values in $\a$ are common for all component errors.

\subsection{Horizontally Interleaved Linearized Reed--Solomon (HILRS) Codes}

We first introduce \ac{LRS} codes~\cite[Definition~31]{martinez2018skew}, which are one of the most prominent families of sum-rank-metric codes.

\begin{definition}[{\Acl{LRS} Codes}]
	\label{def:LRS_codes}
	Let $\vecxi = (\xi_1, \dots, \xi_{\shots}) \in \Fqm^{\shots}$ contain elements of distinct nontrivial conjugacy classes of $\Fqm$.
	Further denote by $\n = (n_1, \dots, n_{\shots}) \in \NN^{\shots}$ a length partition of $n$, i.e., $n = \sum_{i=1}^{\shots} n_i$.
	Let the vectors $\shot{\vecbeta}{i} = (\subShot{\beta}{1}{i}, \dots, \subShot{\beta}{n_i}{i}) \in \Fqm^{n_i}$ contain $\Fq$-linearly independent $\Fqm$-elements for all $i = 1, \dots, \shots$ and write $\vecbeta = \left(\shot{\vecbeta}{1} \mid \dots \mid \shot{\vecbeta}{\shots} \right) \in \Fqm^n$.
	A \acf{LRS} code of length $n$ and dimension $k$ is defined as
	\begin{equation}
		\linRS{\vecbeta}{\vecxi}{\n}{k}
		=\left\{
		\left(
		\opev{f}{\shot{\vecbeta}{1}}{\xi_1} \mid \dots \mid \opev{f}{\shot{\vecbeta}{\shots}}{\xi_\shots}
		\right) : f \in \SkewPolyring_{<k} \right\} \subseteq \Fqm^{n}.
	\end{equation}
\end{definition}
Every codeword $\c \in \linRS{\vecbeta}{\vecxi}{\n}{k}$ corresponds to a skew polynomial $f \in \allowbreak\SkewPolyring_{< k}$.
We sometimes write $\c = \c(f)$ to emphasize this and call $f$ the \emph{message polynomial} of $\c$.

The minimum distance $d$ of an \ac{LRS} code satisfies the Singleton-like bound $d \leq n - k + 1$ with equality.
Thus, \ac{LRS} codes are~\ac{MSRD} codes.

Similar to \ac{RS} and Gabidulin codes, \ac{LRS} codes have a generator matrix $\G$ of a particularly useful form.
Namely, the matrix $\G = (\shot{\G}{1} \mid \dots \mid \shot{\G}{\shots}) = \opMoore{k}{\vecbeta}{\vecxi} \in \Fqm^{k \times n}$ with
\begin{equation}
	\shot{\G}{i} = \opVandermonde{k}{\shot{\vecbeta}{i}}{\xi_i} =
	\begin{pmatrix}
		\subShot{\beta}{1}{i} & \dots & \subShot{\beta}{n_i}{i}
		\\
		\op{\xi_i}{\subShot{\beta}{1}{i}} & \dots & \op{\xi_i}{\subShot{\beta}{n_i}{i}}
		\\
		\vdots & \ddots & \vdots
		\\
		\opexp{\xi_i}{\subShot{\beta}{1}{i}}{k-1} & \dots & \opexp{\xi_i}{\subShot{\beta}{n_i}{i}}{k-1}
	\end{pmatrix}
	\in \Fqm^{k\times n_i}
\end{equation}
for all $i = 1, \dots, \shots$ generates the code $\linRS{\vecbeta}{\vecxi}{\n}{k}$.
\\

We obtain an \ac{HILRS} code with interleaving order $\intOrder \in \NN^{\ast}$ by combining $\intOrder$ \ac{LRS} component codes.
Namely, each codeword of the \ac{HILRS} code is the horizontal concatenation of $s$ codewords of the chosen component codes.

\begin{definition}[Horizontally Interleaved LRS Codes]
	\label{def:HILRS_code}
	Fix an interleaving order $\intOrder \in \NN^{\ast}$ and pick for each $j = 1, \dots, \shots$ an \ac{LRS} code $\linRS{\vecbeta_j}{\vecxi}{\n}{k}$ according to~\autoref{def:LRS_codes}.
	We define the \acf{HILRS} code with interleaving order $\intOrder$, code locators $\vecbeta \defeq (\vecbeta_1 \mid \dots \mid \vecbeta_{\intOrder})$, evaluation parameters $\vecxi$, and length partition $\intOrder \n \defeq (\intOrder n_1, \dots, \intOrder n_\shots)$ as
	\begin{equation}
		\horIntLinRS{\vecbeta}{\vecxi}{\intOrder}{\intOrder\n}{\intOrder k}
		= \left\{\left(\c_1 \mid \dots \mid \c_{\intOrder}\right)
		: \begin{array}{c}
			  \c_j \in \linRS{\vecbeta_j}{\vecxi}{\n}{k} \\ \text{for all } j = 1, \dots, \intOrder
		\end{array}\right\} \subseteq \Fqm^{\intOrder n}.
	\end{equation}
\end{definition}

The code $\horIntLinRS{\vecbeta}{\vecxi}{\intOrder}{\intOrder \n}{\intOrder k}$ has length $\intOrder n$ and dimension $\intOrder k$ over $\Fqm$.
Its minimum distance $d$ equals the minimum distance of its component codes, i.e., $d = n - k + 1$.
\ac{HILRS} codes are hence \emph{not} \ac{MSRD}.
Similar to \ac{LRS} codes, we write $\c(\f) = (\c_1(f_1) \mid \dots \mid \c_{\intOrder}(f_{\intOrder}))\in \horIntLinRS{\vecbeta}{\vecxi}{\intOrder}{\intOrder \n}{\intOrder k}$ with $\f = (f_1, \dots, f_\intOrder)$ and $f_j \in \SkewPolyring_{<k}$ for each $j = 1, \dots, \intOrder$ to emphasize the relation to the message polynomials of the component codewords $\c_1, \dots, \c_{\intOrder}$.
We call $\f$ the \emph{message-polynomial vector} corresponding to $\c$.

\begin{remark}
    It is straightforward to generalize~\autoref{def:HILRS_code} and all concepts of this paper to component codes with different length partitions, lengths, and dimensions.
	However, we assume that the component codes only have different code locators $\vecbeta_j$ for $j = 1, \dots, \intOrder$ for simplicity of notation.
	\qed
\end{remark}

\section{A Gao-like Decoder for HILRS Codes}
\label{sec:gao-dec}

We now derive a Gao-like decoder in the spirit of~\cite{gao2003new,wachter2013phd,puchinger2017row} for \ac{HILRS} codes and the interleaved sum-rank-channel model described in~\eqref{eq:sum_rank_channel_hor_int}.
Let $\y = \c + \e \in \Fqm^{\intOrder n}$ denote the received vector after the codeword $\c = \c(\f) \in  \horIntLinRS{\vecbeta}{\vecxi}{\intOrder}{\intOrder \n}{\intOrder k}$ was corrupted by the error $\e \in \Fqm^{\intOrder n}$ of sum-rank weight $\SumRankWeight(\e) = t$ during transmission.
Recall that we assume a uniform error distribution, that is, that $\e$ is chosen uniformly at random from the set of all vectors of sum-rank weight $t$ as given in~\eqref{eq:e-set-random}.

The main ingredient of the decoder is the Gao-like key equation that exploits the relation between certain polynomials to recover the error values as zeros of the error-span polynomial.
Then, the message-polynomial vector $\f$ that corresponds to $\c$ can be retrieved.
\\

The \emph{\ac{ESP}} $\ESP\in\SkewPolyring$ makes use of the error decomposition shown in~\eqref{eq:err_int_vec_decomp}.
It is the skew polynomial that vanishes at all error values, i.e.,
\begin{equation}
	\opev{\ESP}{\shot{\a}{i}}{\xi_i} = \0 \qquad \text{for all } i = 1, \dots, \shots.
\end{equation}
For horizontal interleaving, the component errors $\e_j$ share the same error values $\a$ for all $j = 1, \dots, \intOrder$ according to~\eqref{eq:err_int_vec_decomp}.
This implies that the \ac{ESP} is common for all component errors.

Next let $G_j \in \SkewPolyring$ for each $j = 1, \dots, \intOrder$ be the minimal skew polynomial for the code locators $\vecbeta_j$ with respect to generalized operator evaluation.
Namely,
\begin{equation}
	G_j(x) \defeq \mpolArgs{\vecbeta_j}{\vecxi}(x) \qquad \text{for all } j = 1, \dots, \intOrder.
\end{equation}
Remark that these polynomials only depend on code parameters and can thus be precomputed.
Further, define $R_j \in \SkewPolyring$ for each $j = 1, \dots, \intOrder$ as the interpolation polynomial whose evaluation at the code locators $\vecbeta_j$ yields the channel observation $\y_j$.
That means that $R_j(x) \defeq \intpolArgs{\vecbeta_j}{\vecxi}{\y_j}(x)$ satisfies
\begin{equation}
	\opev{R_j}{\vecbeta_j}{\vecxi} = \y_j \qquad \text{for all } j = 1, \dots, \intOrder.
\end{equation}
Note that the polynomials $R_j$ can be computed directly from the channel observation $\y = (\y_1 \mid \dots \mid \y_{\intOrder})$.

\begin{theorem}[Gao-like Key Equation for HILRS Codes]
	\label{thm:gao_key_equation_hilrs}
	Let $\c = \c(\f) \in \horIntLinRS{\vecbeta}{\vecxi}{\intOrder}{\intOrder \n}{\intOrder k}$ be a codeword corresponding to the message-polynomial vector $\f = (f_1, \dots, f_{\intOrder})$ with $f_j \in \SkewPolyring_{< k}$ for all $j = 1, \dots, \intOrder$.
	Let further $\y = \c + \e \in \Fqm^{\intOrder n}$ denote a channel observation according to~\eqref{eq:sum_rank_channel_hor_int}.
	For the \ac{ESP} $\ESP \in \SkewPolyring$ and the polynomials
	\begin{align}
		G_j(x) = \mpolArgs{\vecbeta_j}{\vecxi}(x) \quad \text{and} \quad
		R_j(x) = \intpolArgs{\vecbeta_j}{\vecxi}{\y_j}(x) \qquad \text{for each } j = 1, \dots, \intOrder,
	\end{align}
	it holds
	\begin{equation}
		\label{eq:key_equation_gao_hilrs}
		\ESP \cdot R_j \equiv \ESP \cdot f_j \modr G_j
		\qquad \text{for all } j = 1, \dots, \intOrder.
	\end{equation}
\end{theorem}

\begin{proof}
	Consider a fixed $j = 1, \dots, \intOrder$ and let us show the equivalent formulation
	\begin{equation}
		\ESP \cdot (R_j - f_j) \equiv 0 \modr G_j
	\end{equation}
	of the key equation.
	By definition, we know that the evaluation of $R_j - f_j$ at $\vecbeta_j$ is $\opev{(R_j - f_j)}{\vecbeta_j}{\vecxi} = \y_j - \c_j = \e_j$.
	Thus,
	\begin{align}
		\opev{(\ESP \cdot (R_j - f_j))}{\vecbeta_j}{\vecxi}
		\overset{(\triangle)}{=} \opev{\ESP}{\opev{(R_j - f_j)}{\vecbeta_j}{\vecxi}}{\vecxi}
		= \opev{\ESP}{\e_j}{\vecxi}
		= \0
	\end{align}
	applies, where $(\triangle)$ follows from the product rule for generalized operator evaluation and the other equalities hold by definition.
	Together with the fact that $G_j$ is the minimal polynomial of the code locators, we conclude that $G_j$ divides $\ESP \cdot (R_j - f_j)$ on the right.
	Since this argument is true for every $j = 1, \dots, \intOrder$, the statement follows.
	\qed
\end{proof}

As can be seen from the proof of~\autoref{thm:gao_key_equation_hilrs}, the Gao-like key equation~\eqref{eq:key_equation_gao_hilrs} is in fact equivalent to
\begin{equation}\label{eq:gao-key-equation-evaluation-equivalent}
	\opev{\left(\ESP \cdot (R_j - f_j) \right)}{\vecbeta_j}{\vecxi} = \0 \qquad \text{for all } j = 1, \dots, \intOrder.
\end{equation}
By rewriting it in terms of a system of $\Fqm$-linear equations, we obtain
\begin{equation}
	\label{eq:gao_linear_system}
	\underbrace{
	\begin{pmatrix}
		\left( \opMoore{t + k}{\vecbeta_1}{\vecxi} \right)^{\top} & & & - \left( \opMoore{t + 1}{\y_1}{\vecxi} \right)^{\top} \\
		& \ddots & & \vdots \\
		& & \left( \opMoore{t + k}{\vecbeta_{\intOrder}}{\vecxi} \right)^{\top} & - \left( \opMoore{t + 1}{\y_{\intOrder}}{\vecxi} \right)^{\top}
	\end{pmatrix}
	}_{=: \M^{\top}}
    \cdot
	\begin{pmatrix}
	    \boldsymbol{\ESP} \f_1 \\
		\vdots \\
		\boldsymbol{\ESP} \f_{\intOrder} \\
		\boldsymbol{\ESP}
	\end{pmatrix}
	= \0.
\end{equation}
Here, the vectors $\boldsymbol{\ESP}$ and $\boldsymbol{\ESP} \f_j$ for $j = 1, \dots, \intOrder$ contain the coefficients of the respective polynomials, i.e.,
\begin{align}
    (\boldsymbol{\ESP} \f_j)^{\top} &\defeq \left( (\ESP \cdot f_j)_1, \dots, (\ESP \cdot f_j)_{t + k} \right) \in \Fqm^{t + k}
	\quad \text{for all } j = 1, \dots, \intOrder \\
	\text{and } \qquad \boldsymbol{\ESP}^{\top} &\defeq \left( \ESP_1, \dots, \ESP_{t + 1} \right) \in \Fqm^{t + 1}.
\end{align}
Equation~\eqref{eq:gao_linear_system} displays a homogeneous system of $\intOrder n$ equations in $\intOrder (t + k) + t + 1 = (\intOrder + 1) t + \intOrder k + 1$ unknowns.
It can be solved by Gaussian elimination with a complexity of $\OCompl{\max(\intOrder n, (\intOrder + 1) t + \intOrder k + 1)^{\omega}}$ operations in $\Fqm$~\cite[Proposition~2.15.]{Storjohann2000}.
\\

As soon as the Gao-like key equation is solved, we have access to a candidate $\tilde{\ESP}$ for the \ac{ESP} $\ESP \in \SkewPolyring$ as well as to candidates $p_j$ for the products $\ESP \cdot f_j \in \SkewPolyring_{< t + k}$ for all $j = 1, \dots, \intOrder$.
Thus, for any $j = 1, \dots, \intOrder$, left division of $p_j$ by $\tilde{\ESP}$ recovers a candidate $\tilde{f}_j$ for the $j$-th message polynomial $f_j$.
If the remainder $r_j$ of the left division of $p_j$ by $\tilde{\ESP}$ is nonzero for any $j = 1, \dots, \intOrder$ or if any of the $\tilde{f}_1, \dots, \tilde{f}_{\intOrder}$ has degree at least $k$, we declare a decoding failure.
Otherwise, the decoding was correct and $\tilde{f}_j = f_j$ applies for all $j = 1, \dots, \intOrder$.
\autoref{alg:decoder} summarizes all steps of the Gao-like decoder.
\\

\begin{algorithm}
    \caption{Gao-like Decoder for \ac{HILRS} Codes}
    \label{alg:decoder}

    \Input{%
		received vector $\y \in \Fqm^{\intOrder n}$ with $\y = \c(\f) + \e$ according to~\eqref{eq:sum_rank_channel_hor_int} and with $\c(\f) \in \horIntLinRS{\vecbeta}{\vecxi}{\intOrder}{\intOrder \n}{\intOrder k}$\newline{}
		precomputed $G_1, \dots, G_{\intOrder}$ with $G_j \defeq \mpolArgs{\vecbeta_j}{\vecxi}(x)$ for all $j = 1, \dots, \intOrder$
	}
    \Output{$\f = (f_1, \dots, f_{\intOrder})$ or "decoding failure"}

	$R_j \defeq \intpolArgs{\vecbeta_j}{\vecxi}{\y_j}(x) \in \SkewPolyring$ for all $j = 1, \dots, \intOrder$\;

	\tcc{use $\ESP \cdot R_j \equiv \ESP \cdot f_j \modr G_j$ to find $p_j \triangleq \ESP \cdot f_j$ and $\tilde{\ESP} \triangleq \ESP$}

	$(p_1, \dots, p_{\intOrder}, \tilde{\ESP}) \defeq \keyeq(R_1, \dots, R_{\intOrder}, G_1, \dots, G_{\intOrder}, n, k, \intOrder)$\;

	\ForAll{$j = 1, \dots, \intOrder$}{
		$(\tilde{f}_j, r_j) \defeq \ldiv(p_j, \tilde{\ESP})$\;
		\If{$r_j \neq 0$ \textup{\textbf{or}} $\deg(\tilde{f}_j) \geq k$}{
			\Return "decoding failure"
		}
	}

	\Return $\f \defeq (\tilde{f}_1, \dots, \tilde{f}_{\intOrder})$
\end{algorithm}

Let us now further investigate the structure of $\M^{\top}$, which gives rise to the decoding-failure probability $\failProb$.
Remark that the system~\eqref{eq:gao_linear_system} has a nontrivial solution by definition, which implies $\rk_{q^m}(\M) \leq (\intOrder + 1) t + \intOrder k$.
Moreover, a decoding failure can only occur if the solution space of~\eqref{eq:gao_linear_system} has dimension greater than one.
In other words, $\rk_{q^m}(\M^{\top}) = \rk_{q^m}(\M) < (\intOrder + 1) t + \intOrder k$ must apply and we obtain the inequality
\begin{equation}
	\label{eq:failure_prob_and_rank_M}
	\failProb \leq \Pr\left( \rk_{q^m}(\M) < (\intOrder + 1) t + \intOrder k) \right).
\end{equation}
The following lemma gives a characterization of when the solution space of~\eqref{eq:gao_linear_system} is one-dimensional.
Recall that this case implies correct decoding.
\begin{lemma}
	\label{lem:equivalent-failure-bound}
	Consider a vector $\y = \c + \e \in \Fqm^{\intOrder n}$ that was received after transmitting $\c \in \horIntLinRS{\vecbeta}{\vecxi}{\intOrder}{\intOrder \n}{\intOrder k}$ over the channel~\eqref{eq:sum_rank_channel_hor_int}.
	Assume that the error has weight $\SumRankWeight(\e) = t \leq n - k$ and can be decomposed into $\e = \a \cdot \B$ according to~\eqref{eq:err_int_vec_decomp}.
	Further, define $\M$ as in~\eqref{eq:gao_linear_system} and let $\H = \diag(\H_1, \dots, \H_\intOrder)\in \Fqm^{\intOrder (n-k-t) \times \intOrder n}$ be a parity-check matrix of the code $\horIntLinRS{\vecbeta}{\vecxi}{\intOrder}{\intOrder \n}{\intOrder(k + t)}$.
	Then,
    \begin{equation}
		\rk_{q^m}(\M) = (\intOrder + 1) t + \intOrder k
		\qquad \text{ if and only if } \qquad
		\rk_{q^m}(\B \H^{\top}) = t.
	\end{equation}
\end{lemma}

\begin{proof}
	First note that the upper part of $\M$ is a generator matrix of the code $\horIntLinRS{\vecbeta}{\vecxi}{\intOrder}{\allowbreak\intOrder \n}{\intOrder(k + t)}$.
	In other words, the $j$-th block on its diagonal generates $\linRS{\vecbeta_j}{\vecxi}{\n}{\allowbreak k + t}$ for all $j = 1, \dots, \intOrder$.
	For any $j = 1, \dots, \intOrder$, the additivity of the generalized operator evaluation yields $\opMoore{t + 1}{\y_j}{\vecxi} = \opMoore{t + 1}{\c_j}{\vecxi} + \opMoore{t + 1}{\e_j}{\vecxi}$.
	Further, $\c_j \in \linRS{\vecbeta_j}{\vecxi}{\n}{k} = \RowspaceFqm{\opMoore{k}{\vecbeta_j}{\vecxi}}$ implies $\opexp{\vecxi}{\c_j}{\iota} \in \RowspaceFqm{\opMoore{k + \iota}{\vecbeta_j}{\vecxi}}$ for all $\iota = 1, \dots, t$.
	We can hence consider the matrix
	\begin{equation}\label{eq:equivalent_M}
		\widetilde{\M} =
		\begin{pmatrix}
			\opMoore{t + k}{\vecbeta_1}{\vecxi} & & \\
			& \ddots & \\
			& & \opMoore{t + k}{\vecbeta_{\intOrder}}{\vecxi} \\
			\hline
			\opMoore{t + 1}{\e_1}{\vecxi} & \dots & \opMoore{t + 1}{\e_{\intOrder}}{\vecxi}
		\end{pmatrix}
		=:
		\begin{pmatrix}
			\U \\
			\hline
			\L
		\end{pmatrix}
	\end{equation}
	which has the same $\Fqm$-linear row space, and thus the same $\Fqm$-rank, as $\M$.
	In the following, we denote the upper $\intOrder (t + k)$ rows of $\M$ by $\U$ and the lower part by $\L$ for convenience.
	The error decomposition and the $\Fq$-linearity of the generalized operator evaluation let us write $\L = \opMoore{t + 1}{\a}{\vecxi} \cdot \B$.
	Therefore,
	\begin{equation}
		\widetilde{\M}
		= \left(
		\begin{array}{c|c}
			\I_{\intOrder (t + k)} & \0 \\
			\hline
			\0 & \opMoore{t + 1}{\a}{\vecxi}
		\end{array}
		\right) \cdot
		\begin{pmatrix}
			\U \\
			\hline
			\B
		\end{pmatrix}
	\end{equation}
	applies, where $\I_{\intOrder (t + k)}$ denotes the identity matrix of size $\intOrder (t + k) \times \intOrder (t + k)$.
	Since the left matrix has full column rank over $\Fqm$,~\cite[Theorem~2]{matsaglia1974} yields
	\begin{equation}
		\rk_{q^m}(\widetilde{\M}) =
		\rk_{q^m}
		\begin{pmatrix}
			\U \\
			\hline
			\B
		\end{pmatrix}.
	\end{equation}

	Define $\H \defeq \diag(\H_{1}, \dots, \H_{\intOrder}) \in \Fqm^{\intOrder (n - k - t) \times \intOrder n}$ with $\H_j$ being a parity-check matrix of the code $\linRS{\vecbeta_j}{\vecxi}{\n}{k + t}$ for all $j = 1, \dots, \intOrder$.
	Then, $\H$ is a parity-check matrix of $\horIntLinRS{\vecbeta}{\vecxi}{\intOrder}{\intOrder \n}{\intOrder(k + t)}$ and satisfies $\U \H^{\top} = \0$.
	Since
	\begin{align}
		\rk_{q^m}(\M) &= \rk_{q^m}(\U) + \rk_{q^m}(\B) - \dim_{q^m}(\RowspaceFqm{\U} \cap \RowspaceFqm{\B})
		\\
		&\leq (\intOrder + 1) t + \intOrder k - \dim_{q^m}(\RowspaceFqm{\U} \cap \RowspaceFqm{\B})
	\end{align}
	holds, the equality $\rk_{q^m}(\M) = (\intOrder + 1) t + \intOrder k$ is equivalent to $\RowspaceFqm{\U} \cap \RowspaceFqm{\B} = \{\0\}$ and thus to $\RowspaceFqm{\H}^{\perp} \cap \RowspaceFqm{\B} = \{\0\}$.
	This is equivalent to $\rk_{q^m}(\B \H^{\top}) = t$, which proves the lemma.
    \qed
\end{proof}

This equivalent reformulation gives a condition on the error weight $t$ and thus determines the decoding radius.
In fact, the matrix $\B \H^{\top}$ has $t$ rows and $\intOrder (n - k - t)$ columns and can achieve $\rk_{q^m}(\B \H^{\top}) = t$ only if $t \leq \intOrder (n - k - t)$ applies.
Since we obtain a decoding failure in all other cases, we obtain the necessary condition
\begin{equation}
	\label{eq:dec-radius}
    t \leq \tmax \defeq \frac{\intOrder}{\intOrder + 1} (n - k)
\end{equation}
for successful decoding.

We now focus on the zero-derivation case and derive an upper bound on the probability that $\rk_{q^m}(\B\H^\top) < t$ which will also bound the decoding-failure probability according to~\autoref{lem:equivalent-failure-bound}.
Recall that we can choose $\H$ such that $\H_1, \dots, \H_\intOrder$ are generalized Moore matrices, as the dual of an \ac{LRS} code is again an \ac{LRS} code in the zero-derivation setting~\cite[Theorem~4]{martinez2019reliable}.
For such a choice of $\H$, the product $\B \H^{\top} = (\B_1 \H_{1}^{\top} \mid \dots \mid \B_{\intOrder} \H_{\intOrder}^{\top})$ is the transpose of vertically stacked generalized Moore matrices because $\B = (\B_1 \mid \dots \mid \B_{\intOrder})$ contains only $\Fq$-elements and $\op{\xi}{\cdot}$ is $\Fq$-linear for a fixed $\xi \in \Fqm$.
Namely,
\begin{equation}
    \H \B^{\top} =
	\begin{pmatrix}
		\opMoore{t + k}{\h_1\B_1^\top}{\vecxi} \\
		\dots \\
		\opMoore{t + k}{\h_\intOrder\B_\intOrder^\top}{\vecxi} \\
	\end{pmatrix},
\end{equation}
where $\h_j$ denotes the first row of $\H_j$ for each $j = 1, \dots, \intOrder$.

Further recall that, for a fixed rank partition $\t$, the matrix $\B$ is uniformly distributed among the set of all matrices of a particular form having fixed sum-rank weight as described in~\eqref{eq:B-set-random}.
As $\SumRankWeight{(\h_j)}=n$ applies for every $j=1,\dots,\intOrder$, the $(\intOrder \times t)$-matrix containing the vectors $\h_j\B_{j}^{\top}$ as rows is chosen uniformly at random from all matrices in $\Fqm^{\intOrder \times t}$ with sum-rank weight $t$.
This allows us to apply parts of the proof of~\cite[Lemma~7]{bartz2023fast}.

In the zero-derivation setting, we thus obtain the upper bound
\begin{equation}
	\label{eq:failure-probability}
	\failProb \leq \Pr\left( \rk_{q^m}(\B\H^\top) < t \right)
    \leq \kappa_{q}^{\shots + 1} q^{-m ((\intOrder + 1) (\tmax - t) + 1)}
\end{equation}
on the decoding-failure probability $\failProb$, where $\tmax \defeq \frac{\intOrder}{\intOrder + 1} (n - k)$ and $\kappa_q < 3.5$ is defined as $\kappa_q \defeq \prod\limits_{i} \frac{1}{1 - q^{-i}}$ for any prime power $q$.
\\

We implemented the proposed decoder in SageMath~\cite{sage} and ran a Monte Carlo simulation to heuristically verify the tightness of the upper bound on the decoding-failure probability given in~\eqref{eq:failure-probability}.
Note that the actual failure probability is hard to simulate for reasonable parameter sizes, as even the upper bound decreases exponentially.
To obtain observable results, we chose $\Fqm = \F_{3^8}$, $\Fq = \F_3$, and an \ac{HILRS} code of length $n = 16$ and dimension $k = 4$ with
respect to the Frobenius automorphism.
We considered $\shots = 2$ blocks of the same length, namely $\n = (8, 8)$, interleaving order $s = 3$, and randomly chosen
errors of sum-rank weight $t = \tmax = 9$.
The failure probability that we observed for $100$ Monte Carlo errors is $1.569 \cdot 10^{-4}$ while the bound yields $6.535 \cdot 10^{-3}$.
\\

We finish this section with a summary of the results we have obtained so far and give a complexity analysis of the Gao-like decoder for \ac{HILRS} codes.

\begin{theorem}[Gao-like Decoding of \ac{HILRS} Codes]
	\label{thm:gao-like_decoder}
	Consider the transmission of a codeword $\c \in \horIntLinRS{\vecbeta}{\vecxi}{\intOrder}{\intOrder \n}{\intOrder k}$ over the channel~\eqref{eq:sum_rank_channel_hor_int}.
	Let $\y = \c + \e \in \Fqm^{\intOrder n}$ denote the received word and assume that the error $\e$ has bounded sum-rank weight
	\begin{equation}
		\label{eq:dec-radius-cond}
	    \SumRankWeight(\e) = t \leq \frac{\intOrder}{\intOrder + 1} (n - k).
	\end{equation}
	Then, the Gao-like decoder from~\autoref{alg:decoder} can recover $\c$ with a failure probability $\failProb$ that is bounded by
	\begin{equation}
	    \failProb \leq \kappa_{q}^{\shots + 1} q^{-m ((\intOrder + 1) (\tmax - t) + 1)} < 3.5^{\shots + 1} q^{-m ((\intOrder + 1) (\tmax - t) + 1)}
	\end{equation}
	in the zero-derivation setting.
	If the key equation~\eqref{eq:key_equation_gao_hilrs} is solved via Gaussian elimination in the formulation of~\eqref{eq:gao_linear_system}, the overall complexity of the decoder is in the order of $\OComplTilde{(\intOrder n)^{\omega}} \subseteq \OComplTilde{(\intOrder n)^{2.373}}$ operations in $\Fqm$.
\end{theorem}

\begin{proof}
	The decoding radius and the bound on the failure probability were derived above.
	Let us thus focus on the complexity analysis.
	\begin{itemize}
		\item The computation of a minimal or an interpolation polynomial of degree at most $n$ can be done with complexity $\OComplTilde{\mathcal{M}_{q,m}(n)}$ according to~\cite[Section II.D.]{bartz2021orderBases}, e.g.\ by using the recursive formula~\eqref{eq:mpol-lclm}.
			Thus, the computation of $G_1, \dots, G_{\intOrder}$ and $R_1, \dots, R_{\intOrder}$ takes $\OComplTilde{\intOrder \mathcal{M}_{q,m}(n)}$ operations in $\Fqm$.
		\item Finding the solution of the key equation via Gaussian elimination has complexity $\OCompl{\max(\intOrder n, (\intOrder + 1) t + \intOrder k + 1)^{\omega}}$ as stated above.
				Since equation~\eqref{eq:dec-radius-cond} ensures $\intOrder n \geq (\intOrder + 1) t + \intOrder k + 1$, we obtain $\OCompl{(\intOrder n)^{\omega}}$.
		\item The for-loop runs in $\OComplTilde{\intOrder \mathcal{M}_{q,m}(n)}$ operations in $\Fqm$ because the left division in line 4 has complexity $\OComplTilde{\mathcal{M}_{q,m}(n)}$ for each $j = 1, \dots, \intOrder$ according to~\cite[Section II.D.]{bartz2021orderBases}.
		Checking the conditions for a decoding failure is essentially for free.
	\end{itemize}
	Note that $\OComplTilde{\intOrder \mathcal{M}_{q,m}(n)} \subseteq \OComplTilde{\intOrder n^{\min(\frac{\omega + 1}{2}, 1.635)}} \subseteq \OComplTilde{\intOrder n^{1.635}}$.
	Thus, solving the Gao-like key equation determines the overall complexity of $\OComplTilde{(\intOrder n)^{\omega}}$ operations in $\Fqm$.
	\qed
\end{proof}

\section{A Fast Variant of the Gao-like Decoder for HILRS Codes}
\label{sec:gao-mab}

We now present a fast variant of the decoder from~\autoref{alg:decoder}.
As we have seen in its complexity analysis in the proof of~\autoref{thm:gao-like_decoder}, the complexity-dominating task is the solution of the Gao-like key equation.
Thus, we focus on this problem and obtain a performance gain by reformulating it in terms of minimal approximant bases.

Note that we restrict ourselves to the zero-derivation case in this section, even though the used concepts and algorithms generalize straightforwardly to nonzero derivations.
The reason is that the complexity analysis of algorithms involving skew-polynomial operations with nonzero derivations is more involved and was e.g.\ not conducted for the minimal-approximant-basis algorithm~\cite[Algorithm~5]{bartz2021orderBases} that we use for the speedup.

\subsection{Minimal Approximant Bases}
\label{sec:min-approx-bases}

Let us give some definitions and basic properties of minimal approximant bases.
Note that we will only discuss left/row approximant bases and leave out their right/column counterparts, as we are only concerned with these.

Let $\v \in \ZZ^{a}$ be a shifting vector.
Then, the \emph{$\v$-shifted row degree} of a vector $\b \in \SkewPolyringZeroDer^{a}$ is
\begin{equation*}
    \rdeg_{\v}(\b) \defeq \max_{j = 1, \dots, a} \{\deg(b_{j} + v_j)\}.
\end{equation*}
For $\b \in \SkewPolyringZeroDer^{a} \setminus \{\0\}$ and $\v = (v_1, \dots, v_b) \in \ZZ^{a}$, the \emph{$\v$-pivot index} of $\b$ is the largest index $i \in \{1, \dots, a\}$ with $\deg(b_i) + v_i = \rdeg_{\v}(\b)$.

A matrix $\W \in \SkewPolyringZeroDer^{a \times b}$ with $a \leq b$ is in \emph{$\v$-ordered row weak-Popov form} if the $\v$-pivot indices of its rows are strictly increasing in the row index.

A vector $\b \in \SkewPolyringZeroDer^{a}$ is a \emph{left approximant of order $d \in \NN$} of a matrix
$\W \in \SkewPolyringZeroDer^{a \times b}$ if
\begin{equation*}
    \b \W \equiv \0 \modr x^d.
\end{equation*}

A \emph{left $\v$-ordered weak-Popov approximant basis of $\A$ of order $d \in \NN$} is a full-rank matrix $\B \in \SkewPolyringZeroDer^{a \times a}$
in $\v$-ordered row weak-Popov form whose rows are a basis of all left approximants of $\A$ of order $d$.

\subsection{Solving the Gao-like Key Equation via Minimal Approximant Bases}
\label{sec:key-eq-min-approx-bases}

The Gao-like key equation~\eqref{eq:key_equation_gao_hilrs} can also be written as
\begin{equation}
	\label{eq:key_equation_gao_hilrs_non_mod}
	\ESP\cdot f_j = \chi_j\cdot G_j + \ESP\cdot R_j
	\quad \text{for all } j = 1,\dots,\intOrder,
\end{equation}
where $\chi_j\in\SkewPolyringZeroDer$ exists according to the Euclidean algorithm and has degree at most $k+t$ for each $j=1,\dots,\intOrder$.
Observe that~\eqref{eq:key_equation_gao_hilrs_non_mod} implies that the vector
\begin{equation}
	(\ESP\cdot f_1, \dots, \ESP\cdot f_\intOrder, \ESP,\chi_1, \dots, \chi_\intOrder)
	\in \SkewPolyringZeroDer^{2 \intOrder + 1}
\end{equation}
is in the left kernel of the matrix
\begin{equation}
	\label{eq:def_gao_mat_hilrs}
	\W =
	\begin{pmatrix}
		-\I_{\intOrder}
		\\
		\R
		\\
		\G
	\end{pmatrix}
	\in \SkewPolyringZeroDer^{(2 \intOrder + 1) \times \intOrder}
\end{equation}
where $\R \defeq (R_1, \dots, R_{\intOrder})$ and $\G \defeq \diag(G_1, \dots, G_{\intOrder})$.

The following result based on~\cite[Lemma~21]{bartz2021orderBases} is fundamental for reformulating the Gao-like key equation as a minimal-approximant-bases problem.

\begin{lemma}
	\label{lem:mab-conversion}
	Consider the same setting as in~\autoref{thm:gao-like_decoder} and let $\W$ be defined as in~\eqref{eq:def_gao_mat_hilrs}.
	Further write
	\begin{align}
		\vecrho \defeq (\ESP\cdot f_1, \dots, \ESP\cdot f_\intOrder, \ESP)
		\quad \text{and} \quad
		\vecchi \defeq (\chi_1, \dots, \chi_\intOrder)
	\end{align}
	for simplicity.
	Further define the shifting vectors $\w \defeq (\0_{\intOrder}, k-1) \in \ZZ^{\intOrder + 1}$ and $\v \defeq (\0_{\intOrder}, k-1, \0_{\intOrder}) \in \ZZ^{2 \intOrder + 1}$,
	as well as the degree constraints $D \defeq \tmax = \frac{\intOrder}{\intOrder + 1} (n - k)$
	and $d \defeq \degConstraint + n$.
	Then,
	\begin{equation}
		\label{eq:mab-lemma-normal}
		(\vecrho \mid \vecchi) \cdot \W = \0 \quad \text{and} \quad \rdeg_{\w}(\vecrho) < D
	\end{equation}
	if and only if
	\begin{equation}
		\label{eq:mab-lemma-mab}
		(\vecrho \mid \vecchi) \cdot \W \equiv \0 \modr x^d \quad \text{and} \quad \rdeg_{\v}(\vecrho \mid \vecchi) < D.
	\end{equation}
\end{lemma}

\begin{proof}
	We start with showing that~\eqref{eq:mab-lemma-normal} implies~\eqref{eq:mab-lemma-mab}.
	The left-hand side of~\eqref{eq:mab-lemma-mab} clearly follows from~\eqref{eq:mab-lemma-normal} and it remains to show that $\deg(\chi_j) < D$ holds for all $j = 1, \dots, \intOrder$.
	With~\eqref{eq:key_equation_gao_hilrs_non_mod}, we get
	\begin{align}
	    \deg(\chi_j) &\leq \max\{ \deg(\ESP \cdot f_j), \deg(\ESP \cdot R_j) \} - \deg(G_j)
		\\
		&\leq \max\{ t + k - 1, t + n - 1 \} - n
		< t \leq \tmax = D.
	\end{align}

	For the other implication, note that the right-hand side of~\eqref{eq:mab-lemma-mab} directly implies the right-hand side of~\eqref{eq:mab-lemma-normal}.
	In order to see that the left-hand side of~\eqref{eq:mab-lemma-normal} holds, we show that all entries of the vector $(\vecrho \mid \vecchi) \cdot \W$ have degree less than $d$.
	With the help of the right-hand side of~\eqref{eq:mab-lemma-mab} and~\eqref{eq:key_equation_gao_hilrs_non_mod}, we obtain:
	\begin{itemize}
		\item $\deg(\ESP \cdot f_j) < D < d$,
		\item $\deg(\ESP \cdot R_j) \leq \deg(\ESP) + \deg(R_j) \leq t + n - 1 = D + n - 1 < d$,
		\item $\deg(\chi_j \cdot G_j) < t + n = D + n = d$.
	\end{itemize}
	\qed
\end{proof}

Hence, we can solve the Gao-like key equation~\eqref{eq:key_equation_gao_hilrs} by computing a left $\v$-ordered weak-Popov approximant basis $\B$ of $\W$.
This can be accomplished by~\cite[Algorithm~5]{bartz2021orderBases} requiring $\OComplTilde{\OMul{n}} \subseteq \OComplTilde{n^{\min\left\{\frac{\omega+1}{2},1.635\right\}}} \subseteq \OComplTilde{n^{1.635}}$ operations in $\Fqm$.

We then obtain candidates $p_j$ for the products $\ESP \cdot f_j$ for each $j = 1, \dots, \intOrder$ and a candidate $\tilde{\ESP}$ for the $\ESP$ by choosing the row $\bmin$ of $\B$ having minimal $\v$-weighted degree.
This choice makes sure to satisfy the degree constraint in~\eqref{eq:mab-lemma-mab} to get a proper solution as described in~\autoref{lem:mab-conversion}.
The subroutine for solving the Gao-like key equation via the presented minimal-approximant-bases approach is summarized in~\autoref{alg:key_equation_min_approx_basis}.

\begin{algorithm}
	\caption{Subroutine $\keyeqapprox(\cdot)$ for Solving the Gao-like Key Equation via a Minimal Approximant Basis}
	\label{alg:key_equation_min_approx_basis}

	\Input{$R_1, \dots, R_{\intOrder}, G_1, \dots, G_{\intOrder}, n, k, \intOrder$}
	\Output{$p_1, \dots, p_{\intOrder}, \tilde{\ESP}$}

	$\v \defeq (\0_{\intOrder}, k - 1, \0_{\intOrder})$\;

	$\degConstraint \defeq \frac{\intOrder}{\intOrder + 1} (n - k)$ and $d \defeq \degConstraint + n$\;

	$\W \defeq
	\begin{pmatrix}
		-\I_{\intOrder}
		\\
		\R
		\\
		\G
	\end{pmatrix}
	\in \SkewPolyringZeroDer^{(2 \intOrder + 1) \times \intOrder}$ with $\R \defeq (R_1, \dots, R_{\intOrder})$ and $\G \defeq \diag(G_1, \dots, G_{\intOrder})$\;

	\tcc{left $\v$-ordered weak Popov approximant basis of $\W$ of order $d$}
	$\B \defeq \lapproxbasis(d, \W, \v) \in \SkewPolyringZeroDer^{(2 \intOrder + 1) \times (2 \intOrder + 1)}$\;

	Define $\b_{\min} = (b_{\min, 1}, \dots, b_{\min, 2 \intOrder + 1})$ as the minimal row of $\B$ with respect to the $\v$-weighted degree\;

	\Return $b_{\min, 1}, \dots, b_{\min, \intOrder}, b_{\min, \intOrder + 1}$
\end{algorithm}

\begin{theorem}
    \autoref{alg:key_equation_min_approx_basis} solves the Gao-like key equation~\eqref{eq:key_equation_gao_hilrs} in $\OComplTilde{\intOrder^{\omega} n^{1.635}} \subseteq \OComplTilde{\intOrder^{2.373} n^{1.635}}$ $\Fqm$-operations.
\end{theorem}

\begin{proof}
    The complexity of~\autoref{alg:key_equation_min_approx_basis} is dominated by finding a minimal approximant basis in line 4.
	This can be achieved using~\cite[Algorithm~5]{bartz2021orderBases} whose complexity is $\OComplTilde{\intOrder^{\omega} n^{1.635}} \subseteq \OComplTilde{\intOrder^{2.373} n^{1.635}}$~\cite[Theorem~11]{bartz2021orderBases}.
\end{proof}

This directly implies the following complexity improvement for~\autoref{thm:gao-like_decoder}:

\begin{corollary}
	\label{cor:fast-gao}
    When the Gao-like key equation~\eqref{eq:key_equation_gao_hilrs} is solved by~\autoref{alg:key_equation_min_approx_basis}, the complexity of the Gao-like decoder from~\autoref{alg:decoder} decreases to $\OComplTilde{\intOrder^{\omega} n^{1.635}} \subseteq \OComplTilde{\intOrder^{2.373} n^{1.635}}$ operations in $\Fqm$.
\end{corollary}

With~\autoref{cor:fast-gao}, the Gao-like decoder is the fastest known decoder for \ac{HILRS} codes in the sum-rank metric as well as for horizontally interleaved Gabidulin codes in the rank metric.
Its complexity is essentially subquadratic in the component-code length $n$, as the interleaving order $s$ is usually much smaller than the code length $n$.
Remark in particular that the gain in the error-correcting capacity increases fast for increasing $\intOrder$, as $\frac{\intOrder}{\intOrder + 1}$ quickly tends to one.

\section{Conclusion}
\label{sec:conclusion}

We studied \ac{HILRS} codes and their fast decoding which has promising potential applications in code-based cryptography.
As a starting point, we presented a Gao-like decoder that features probabilistic unique decoding for an error of sum-rank weight at most $\frac{\intOrder}{\intOrder + 1} (n - k)$, where $\intOrder$ is the interleaving order, and $n$ and $k$ are the length and the dimension of the component codes.
We gave a bound on the failure probability and achieved a complexity of $\OComplTilde{(\intOrder n)^{2.373}}$ operations in $\Fqm$ by solving the Gao-like key equation conventionally via Gaussian elimination.

Techniques from the area of minimal approximant bases allowed us to speed up the decoder significantly and obtain a complexity of $\OComplTilde{\intOrder^{2.373} n^{1.635}}$ operations in $\Fqm$.
Under the reasonable assumption that the interleaving order $\intOrder$ is small compared to the component-code length $n$, this is subquadratic.
Overall, this results in the fastest known decoders for both \ac{HILRS} codes in the sum-rank metric and for horizontally interleaved Gabidulin codes in the rank metric.
\\

Further work can include the generalization of the presented decoder to the error-erasure case.
Next to errors, this error model includes row and column erasures, for which either the row space or the column space is known.
Moreover, other techniques could give bounds on the failure probability for nonzero derivations or yield tighter ones for the zero-derivation setting.

\bibliographystyle{splncs04}
\bibliography{references}

\end{document}